\documentclass{acm_proc_article-sp}
\usepackage{color}
\usepackage{amsmath}
\usepackage{amsfonts}
\usepackage{amssymb}
\usepackage{caption}
\usepackage{subfig}
\numberwithin{equation}{section}
\newtheorem{theorem}{Theorem}

\begin{document}

\title{Widescope - A social platform for serious conversations on the Web
\titlenote{\texttt{http://widescope.stanford.edu}}}
%
%
%
%
%

\numberofauthors{6} 
%
\author{
%
%
\alignauthor
Noah Burbank\\
       \affaddr{Stanford University}\\
       \email{nburbank@stanford.edu}
\alignauthor
Debojyoti Dutta\titlenote{The work was done while the author was visiting Stanford}\\
       \affaddr{Cisco Systems}\\
       \email{ddutta@gmail.com}
\alignauthor Ashish Goel\\
       \affaddr{Stanford University}\\
       \email{ashishg@stanford.edu}\\
\and  
\alignauthor David Lee\\
       \affaddr{Stanford University}\\
       \email{davidtlee@stanford.edu}\\
\alignauthor Eli Marschner\\
      \affaddr{Stanford Univerisity}\\
       \email{eli@cs.stanford.edu}\\
\alignauthor Narayanan Shivakumar\\
       \affaddr{Stanford University}\\
       \email{shivak@gmail.com}
}
\date{7 November 2011}

\maketitle
\begin{abstract}

There are several web platforms that people use to interact and exchange ideas, such as social networks like Facebook~\cite{facebook}, Twitter~\cite{twitter}, and Google+~\cite{googleplus}; Q\&A sites like Quora and Yahoo! Answers; and myriad independent fora. However, there is a scarcity of platforms that facilitate discussion of complex subjects where people with divergent views can easily rationalize their points of view using a shared knowledge base, and leverage it towards shared objectives, e.g. to arrive at a mutually acceptable compromise.

In this paper, as a first step, we present Widescope, a novel collaborative web platform for catalyzing shared understanding of the US Federal and State budget debates in order to help users reach data-driven consensus about the complex issues involved. It aggregates disparate sources of financial data from different budgets (i.e. from past, present, and proposed) and presents a unified interface using interactive visualizations. It leverages distributed collaboration to encourage exploration of ideas and debate. Users can propose budgets ab-initio, support existing proposals, compare between different budgets, and collaborate with others in real time.

We hypothesize that such a platform can be useful in bringing people's thoughts and opinions closer. Toward this, we present preliminary evidence from a simple pilot experiment, using triadic voting (which we also formally analyze to show that is better than hot-or-not voting), that 5 out of 6 groups of users with divergent views (conservatives vs liberals) come to a consensus while aiming to halve the deficit using Widescope. We believe that tools like Widescope could have a positive impact on other complex, data-driven social issues.

\end{abstract}


\section{Introduction}
\label{sec:intro}

Social Networking sites have transformed our lives dramatically. Services like Twitter~\cite{twitter}, Google+~\cite{googleplus} and Facebook~\cite{twitter} draw hundreds of millions, who end up spending a significant fraction of their time in some social activity leveraging the features of the particular service~\cite{timespent}. Popular activities include sending a message to a group or tweet, chat with others, groupchat, share status and location updates, and using immersive services like group video chat (e.g. Google+ hangouts). Thus social networking tools are quite effective for basic social communications and information and  thought propagation within one's own social network.

Several activities that we perform in our daily lives are more involved that the aforementioned ones. For example, we debate serious topics within our circles like the economy, social change, globalization and governmental budgets. For these serious conversations, we can use the tools available for all social communication like forums. This might be adequate or even optimal for a great number of topics, especially those which are more qualitative or subjective which do not need a lot of data points for argument. However for some topics like the US Federal Budget (and the current deficit~\cite{federalbudget}), we believe that existing social networking mechanisms are not a good fit due to the {\em dynamic range} of the underlying data and issues stemming from those, as mentioned below:
\begin{itemize}
\item {\bf Data Slicing and Aggregation:} A big problem with discussions about the budget in the current status quo is that debates are centered around isolated topics. This is true not only among the average citizens, but also among politicians. In one given debate, the focus may be on taxes, while in another, the focus might shift to  healthcare. As a result, it is hard to pin down exactly what the effects of a candidates proposals are on the entire budget deficit. Thus we need a solution that encapsulates the entire aggregate of proposals which then can be leveraged by users for applications like proposing one's own budget and social consensus.
\item {\bf Different Baselines:} A second difficulty is that different baselines (such as the one from the Congressional Budget Office~\cite{cbo_options}, the President's Budget~\cite{pres_budget}, the House of Representatives' Budget~\cite{prosperity}, or the People's Budget~\cite{progress}) are used by different budget proposals. A budget baseline is an assumption on what the future will look like given certain policy continuations or lack thereof. When different baselines are used then it becomes far too difficult for non-economists to make comparisons. When the only information given is how much {\em improvement} specific positions are making, the conversations needs to compare changes in apples to changes in apples. Thus we need to separate budget positions from baselines by revealing the final dollar amounts that a specific position implies. 
\end{itemize}

The questions we ask in the paper are centered around novel approaches to enable serious conversations on the web and also to help people achieve shared goals. Here are our simple desiderata for such a platform:
\begin{itemize}
\item Users should be able to express their views around a complex idea that leverages rich and diverse data. For example, in order to express an opinion on the Federal Budget, she should be able to quantify proposals and present a reasoning for the data. Also users should be able to choose a popular proposal from another user or entity.  
\item Users should be able to compare different proposals using both a high level summarized view and as well zoom into individual proposal comments. They should be able to vote and view the top proposals among other proposals.
\item Users should be able to collaborate with others in real time to converge upon a shared consensus proposal.  
\end{itemize}
In addition, we ask that given a platform as described above, we ask what the potential social impact could be. For example, does it really help in online serious conversations? Can we use such a platform to achieve some shared goal around the conversation topic? What is a good way to compare ideas and proposals?

\subsection{Our Contributions}

In this paper, we present a platform called Widescope~\cite{widescope,wsjblurb} that enables online social conversations around the US budgets driven by rich data visualization. The platform has several features that handle dynamic range of the data, allow for comparisons and enable collaboration in real time. Using our platform, users can propose a budget by either entering their own proposals or choosing an existing one. Also, users can compare and vote for different budgets, and visualize their budget overlayed on another proposed budget, in real time.

We then present a very simple model of a proposal where an opinion can be mapped onto a single-dimensional point on a line and analyze common voting schemes~\cite{voting} in the model. We formally show that a triadic voting scheme is better than a naive hot-or-not voting scheme~\cite{hotornot,realhotornot} for picking a  winner among those who have moderate views,  and are likely to help users to reach a concensus~\cite{condorcet}.

Then we describe a pilot user trial which leverages this platform where we design an experiment that enables users with divergent opinions to reach a consensus given a shared goal. In this case our goal was to halve the deficit of the current Federal Budget! We use triadic voting followed by a collaborative step to show that 5 out of 6 groups of the users were able to achieve consensus when paired with a user with divergent views (users had declared whether they were conservatives or liberals).

We believe that this paper is a first step and scratches the surface of the huge opportunities on the web for serious conversations around complex topics. The potential social impact using these platforms could be very significant. Conversations could be designed and structured to help users reach consensus through mutual understanding and a focus on facts and thereby increasing the efficiencies of social decision making.

\subsection{Related Work}

Widescope is based on work that is highly interdisciplinary, and leverages ideas from online social networking, visualization, user interface design, web based system design, social science, psychology and social choice theory. In addition, for the specific use case of the US Federal Budget, we compare our work with the existing mechanisms available to users.

While there are several research directions related to the different aspects that Widescope touches, we are not aware of any platform that does what Widescope does - i.e. enable serious conversations around budgets using a data-driven, social platform. 

\subsubsection{Budgets}

There are three main projects of note that have focused on visualization of governmental budget data and motivated individual expression and interaction with that data. First is the New York Times' Budget Puzzle~\cite{nytimesbudgetpuzzle}, which had a number of mutually exclusive deficit reduction options and a novel grid-based visualization. This website, and its popularity was a source of inspiration for the Widescope project. 

Second is the NextTen's California Budget Challenge~\cite{budgetchallenge}, which allows people to choose from a small number (usually 2-4) of deficit-reducing options in about twenty key categories. We thank the Program Director Sarah Henry for sharing how the Budget Challenge started off as a {\em spreadsheet}, but that focus group experiments lead to greater engagement using the options format, much like the NY Times project. We hope to allow the same detail as the {\em spreadsheet} while encouraging engagement like the NY Times project. 

The third is the Washington State Budget~\cite{washbudget}. This site had a user interaction similar to the NY Times, although with a slightly different graphical output. Mr. Itti provided useful feedback during iterations that lead to our current website. It is not uncommon for the state budget transparency requirements to make it to the  office of management, and budget employees getting excited about visualization tools like Widescope, however none seem to employ the level of quantitative decision making nor the breadth of budget change options that Widescope offers; Widescope offers nearly infinitely many ways to adjust the budget.

\subsubsection{Social Choice}

Widescope is related to research in information design, and social choice theory. Social choice theory~\cite{socialchoice} is a theoretical framework for studying aggregate choices for collective decision making. Although this has been a rich research area for decades, most of the findings have been in the area of so-called {\em impossibility results} which illustrate situations in which it is impossible to aggregate individual opinions into a collective opinion~\cite{arrowimpossible}. Accepting this research, we decided that the voting mechanism was not the problem in the polarization of political debate, it was that people were often talking across one another. Studies dating back to the 1960’s have demonstrated that people process identical information differently with different visual presentations ~\cite{knox62}. Thus we hypothesized that if we could improve the conversation, especially through better visualization of budget data to make processing the data more natural, then the resulting data-driven debate might solve the polarization better than a novel voting mechanism. 

In addition to the rich data visualization, Widescope has schemes for voting~\cite{voting} in order to aggregate individual choice. We show how novel voting schemes like our triadic scheme presented in Section~\ref{sec:triadic} can lead to the probability distribution of the winning proposal to be concentrated around the median proposal, which satisifies the Condorcet Criterion~\cite{condorcet}. 

\subsubsection{Computer Human Interaction}

Our work is related to recent work in the area of {\em Computer Human Interaction (CHI)}. Widescope shares common design elements with Pathfinder~\cite{ui:pathfinder}, an online platform that allows users to collaboratively discuss and analyze data. Users can initiate topics that are then augmented by other users with data and analysis. Unlike Pathfinder, our tool is built specifically for the purpose of discussing budgets, uses detailed visualization by default, and has a real-time online collaboration feature. Voyager and Voyeurs~\cite{ui:voyager} shows that the interplay of asynchronous collaboration and visualization is a powerful influencer for user engagement. Work on CommentSpace~\cite{ui:commentspace} showed how adding lightweight tags and link structure to comments can help groups of users collect and analyze data in a collaborative fashion. We assume more about the structure of our data, and take advantage of that to a greater extent than a generic platform like CommentSpace llows without significant modifications. Because our platform targets a more specific use case it allows more sophisticated UI elements for directly comparing and contrasting specific items of data. 

Our work is inspired by such previous studies. In addition, we have extended these concepts to include ideas like synchronous collaborations and a propose/accept interaction metaphor, which we believe is a key differentiation. We allow users to build their own proposals by picking and choosing among others' proposals, as opposed to requiring that users generate original opinions. This approach has potential to reduce the barrier to entry for new users, and decrease duplicate, {\em noisy} commentary. Furthermore, because proposals are treated as first-class objects within the system it's possible to rank them by metrics like overall popularity, and to track their dispersion through the social networks using the tool.

\section{Widescope Platform}

In this section, we describe our platform in detail and demonstrate
a user interface tailored specifically for enabling serious conversations about the US federal and state budgets. The platform was built using a collection of open source software libraries; most notably: Ruby on Rails~\cite{ror}, Protovis~\cite{protoviz}, and JQuery~\cite{jquery}, which were chosen for their widespread adoption and/or unique capabilities. Our data backend is a replicated MySQL~\cite{mysql} environment optimized for read scalability. All components are hosted on Amazon EC2~\cite{ec2} virtual servers.

Collaborative budgeting is made difficult by confounding
factors such as party affiliation, ignorance of important, hard-to-find facts,
and mis-communication between dissenting participants.
Widescope provides structured channels for interaction and debate, and eschews expression of political
affiliation or similar biases in a deliberate effort to minimize
these factors.

The interface affords:
\begin{enumerate}
\item Specification of proposals for adjusting budget
items to specific dollar amounts
\item Comparing and commenting on those proposals in
a structured, quantified way
\item Mixing-and-matching different proposals into a
personalized collection that amounts to a holistic
proposal for an overall budget
\end{enumerate}

This proposal-centric structure serves to focus and
synchronize users' attention, and provides a framework
for relating different points of view. For example, some
people may not have any idea about how to adjust multiple
items within a specific budget, but they might know a lot
about issues relating to a particular budget item. They are
free within the tool to express their deep knowledge on
a particular subject by creating a comprehensive set of
proposals without being pressured to figure out which other
budget adjustments they complement. Other users can mix
and match such peer-generated proposals.
Affording various levels of interaction such as these allows
the majority of people to have a voice and build shared
understanding without requiring more expertise than they
might actually have -- or having to think about the issues at
hand in an unnatural way.

\subsection{Data Visualization}

\begin{figure}[htb]
\label{fig:ui}
\centering
\includegraphics[width=3in]{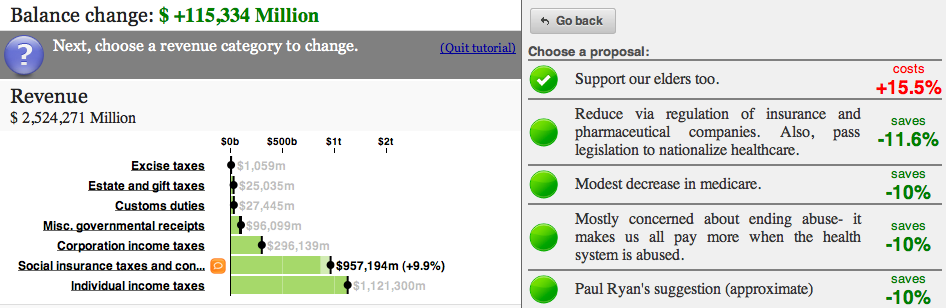}
\caption{UI instrumented with walk-through tutorial prompts (left), showing adjustment suggestions (right).}
\end{figure}

\begin{figure}[htb]
\label{fig:barchart}
\centering
\includegraphics[width=3.25in]{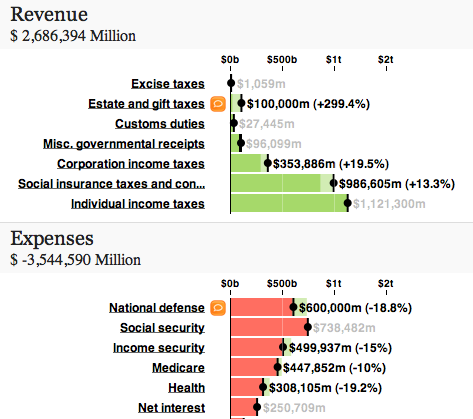}
\caption{Interactive bar chart visualizing budget allocations.}
\end{figure}

\begin{figure}[htb]
\label{fig:proposalform}
\centering
\includegraphics[width=3in]{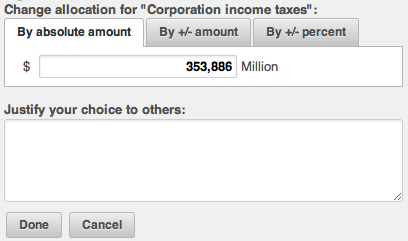}
\caption{Precise budget proposal entry form}
\end{figure}

Widescope's primary feature for expressing budget data is an
interactive bar chart (Figure~\ref{fig:barchart}) that (a) shows the raw dollar amounts
for budget items -- relative to each other and grouped as either an
expense (negative) or revenue (positive) -- and (b) affords inline
adjustment and experimentation via a click-and-drag interaction.

Because users may not know how to begin modifying a budget, 
suggested adjustments are pulled from the pool of existing proposals (Figure~\ref{fig:ui}). If a user
simply wants to learn more about a particular part of a budget there
are category-level descriptive popups that are triggered by hovering
over the labels for budget item bars (e.g. "Health", "National
Defense", "Income Tax"). Fine-grain adjustments are possible
using a form (Figure~\ref{fig:proposalform}) that includes fields for specifying the exact
amount of an adjustment and text for the rationale behind it.

Proposals used by a budget are displayed in a panel next to the interactive bar chart, and comments on individual proposals from other users are displayed next to those. The same panel contains a drop-down menu listing other budgets
that use each proposal. This provides a clue as to how popular a particular proposal is, and a way to navigate to potentially similar budgets. There is a number displayed prominently above the bar charts that represents
the difference, positive or negative, of the adjusted budget relative
to the original budget. This provides an overview metric
indicating how significant all adjustments made so far are to the
bottom line of the budget in question.

When thinking about a budget, users may consider different topics at varying levels of specificity. 
Considerations at a very basic level could include: what budget is being adjusted (e.g. Federal? California?
Nebraska?), and whether or not that budget is relevant to a given user. 
More complex considerations include: how balanced is a budget overall, what is
the balance between certain spending and revenue items within that
budget, and what nuanced reasoning might be behind specific monetary allocations (e.g. a 5\% increase in
income tax revenue, sourced specifically from people earning more than \$250,000
per year).  All but this last, most fine-grained level of specificity
are addressed directly by the bar chart visualization and its
supporting interface elements. The latter, more nuanced kind of consideration, is relegated
to free-form text, either as the content of a proposal's explanation or comments made on a proposal, or both.

\begin{figure}[htb]
\label{fig:commentform}
\centering
\includegraphics[width=3in]{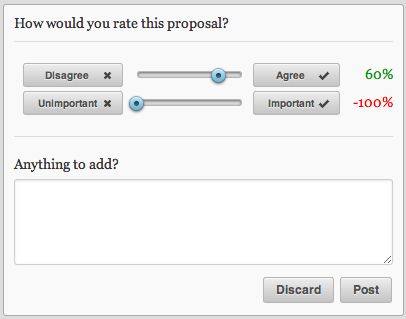}
\caption{Comment entry form with sliders for agreement and importance.}
\end{figure}

\begin{figure}[htb]
\label{fig:comment}
\centering
\includegraphics[width=2.5in]{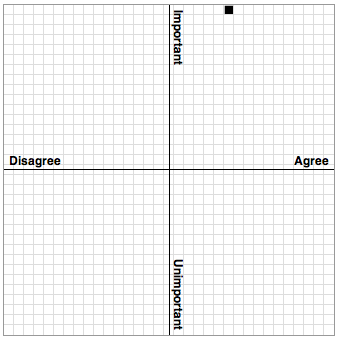}
\caption{Comment agreement and importance matrix.}
\end{figure}

The commenting system invites users to quantify otherwise free-form feedback by indicating amount of agreement or disagreement with a given proposal, as well as perceived importance or the lack thereof (Figure~\ref{fig:commentform}).  This feature is meant to encourage structured
commentary on relevant issues, and to provide concrete numerical data
points for further analysis. This feedback is
represented on a 2 dimensional scatter plot with axes for
agreement/disagreement (horizontal axis), and importance/unimportance (vertical axis) (Figure~\ref{fig:comment}). Precise ratings are aggregated within discrete grid cells, which vary in tone from light grey to black depending on the amount of overlap among aggregated ratings. A cell representing more ratings with appear darker than one representing fewer.
This provides a visually cohesive representation of the
discussion about a particular proposal, and can be used as a high-level
summary of sentiment in case commentary grows too large to be read and thoroughly understood.

\subsection{Voting and Collaboration}

\begin{figure}[htb]
\label{fig:voting}
\centering
\includegraphics[width=3.25in]{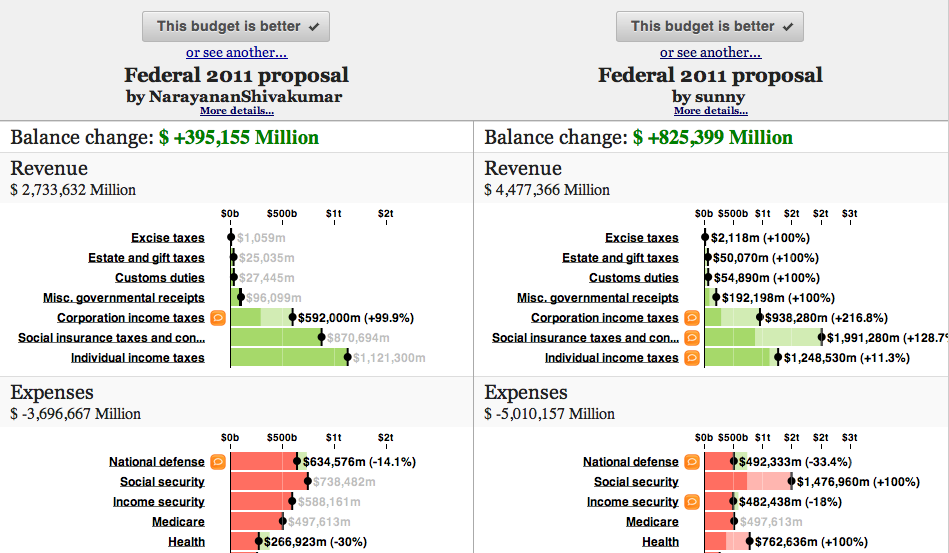}
\caption{Budget comparison and voting page.}
\end{figure}

A significant challenge facing official budget
administrators is that they are faced with their task having little or no relevant prior
experience, especially when freshly elected. Anyone using Widescope for the first time faces a similar challenge: ``Where do I start and what am I
even supposed to do?" The tool has specific features to address this
challenge: a guided walkthrough (Figure~\ref{fig:ui}), which is
the default mode for new users; proposal suggestions (always given prominence over entering new, from-scratch proposals); and comparison voting. The latter works by
showing a user two budgets for the same state side-by side,
and inviting users to select which one is {\em better} (Figure~\ref{fig:voting}).
This serves as a low-friction voting interface, as users don't
have to absolutely agree with a budget to vote for it; they
just have to agree with it more than the other budget. It
provides an easy mode of interaction so that even users without much to contribute to the budget-making process can contribute by curating others' budgets.

\begin{figure}[htb]
\centering
\includegraphics[width=3.25in]{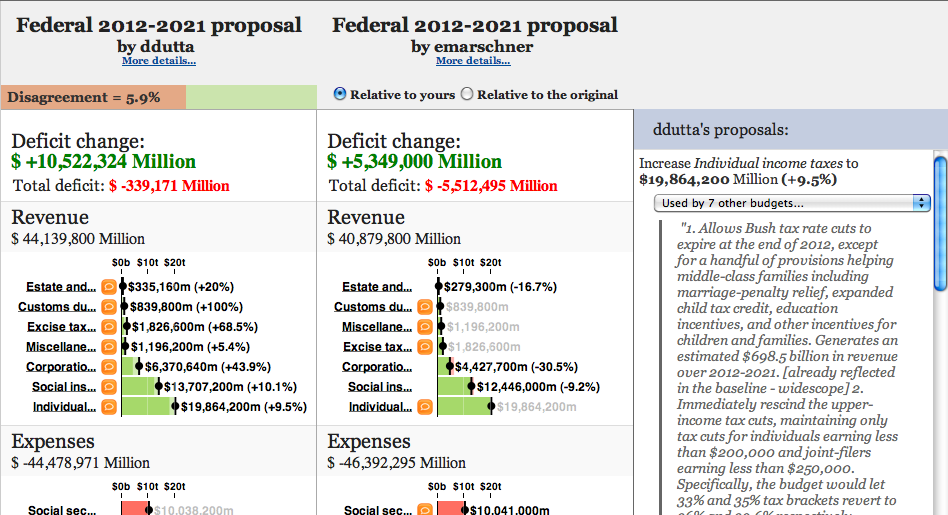}
\caption{The collaboration page. Two users can collaborate in real time. The user's budget is displayed on the left side and the other user\'s budget is overlayed either on top of her budget or on the baseline budget.}
\label{fig:collab}
\end{figure}

Because rapid feedback and personal dialog are key factors for effective debate, Widescope has a specialized interface for real-time collaboration. The collaboration interface enables users who are simultaneously logged in to send continuously updating budget changes to each other via the interface depicted by Figure~\ref{fig:collab}. A user's own budget is displayed with the left-most column of charts, and the other user's budget, represented by the right-most column of charts, can be overlaid either on top of her budget or on the baseline budget. As a user moves her chart's bars and selects various proposals, the behavior of the interface is exactly as if she was editing her own budget, but the chart on the right representing the other user's budget updates simultaneously with her changes (if viewing the other's overlaid on hers), and with any changes made by the other user.  The amount of {\em Disagreement} between each user's budgets is represented as a simple horizontal percentage meter above the charts to the left, where {\em Disagreement} is defined as:

Given users \emph{A} and \emph{B}, Given category \emph{C} $\in$ Revenue and Expense categories
$A_{c}=$ Dollars allocated to \emph{C} by User \emph{A}
$B_{c}=$ Dollars allocated to \emph{C} by User \emph{B}
$\Delta_{c}=A_{c}-B_{c}$

$\text{Disagreement}=\frac{\sum_{c}|\Delta_{c}|}{\sum_{c}\frac{|A_{c}|+|B_{c}|}{2}}$

Any amount of {\em Disagreement} beyond 10\% is considered quite significant and displayed simply as $>$10\%.

Though these mechanisms for collaboration are primitive, they allows two users to visualize each others' budgets in real time, relative to their own budgets, with a choice of overlays. This enables a different kind of comparative budgeting than previously described for times when users may care not only about which of two budgets is {\em better}, but also about how changes to specific categories of one budget relate to those of another, and about having an active conversation about otherwise static data.

\section{Triadic Voting Schemes}
\label{sec:triadic}

The purpose of voting is to take voter opinions over a candidate set and produce a good candidate, as measured by various properties. Given the comparison functionality in Widescope, it is natural to ask if there are voting schemes that can take opinions, expressed as comparisons between budgets, and aggregate them together in a way to choose a good budget. In this section, we propose a comparison-based voting scheme, Triadic Voting, and compare it to a commonly used comparison-based voting scheme, Hot-or-Not. We demonstrate that Triadic Voting is more efficient at eliminating outliers and converging to good solutions compared to the typical Hot-or-Not solution.

{\it Triadic Voting} is an iterative voting scheme that distinguishes between moderate and strongly polarized proposals. We compare this to a iterative version of the Hot-or-Not scheme. In the standard Hot-or-Not scheme, two proposals are randomly presented to a user who selects the one that she likes the most\footnote{Note that Hot-or-Not sometimes refers to an absolute rating of single proposals on some scale. We are referring to the comparison-based variants~\cite{realhotornot}}. Based on several of these votes, we eliminate the losing proposals and repeat. The general idea of an iterative voting scheme is that highly polarized proposals are also generally liked by a smaller portion of the population. When pitted against another moderate proposal, a highly polarized proposal will lose more often than it wins and hopefully, highly polarized proposals can be quickly eliminated.

This can easily seen and formalized if the proposal space is {\it single-peaked}, i.e., the proposals can be approximated by a one-dimensional number. Assume that voter opinions are points on a line and that each voter gives a proposal corresponding to the point it lies on. Any voter will like proposals that are closer to it as per the distance between the two proposals. In this case, the median proposal (which is also intuitively the most moderate proposal) satisfies the {\it Condorcet property}, which means that it would win any pairwise election against another proposal. This is because at least half of the voters will vote for the given proposal. If the opposing proposal is on the right side of it, then the median proposal will win all the voters on the left and vice versa.

In the Triadic voting scheme we propose, three voters $x, y, z$ are randomly selected, each of which has his or her own proposal. Then, each voter is made to vote between the other two proposals. That is, $x$ votes between $y$ and $z$, $y$ votes between $x$ and $z$, and $z$ votes between $x$ and $y$. This is similar to Hot-or-Not in that users are making pairwise comparisons. The difference, however, is that the comparisons being made are correlated in a very precise way depending on the group of three candidates that are chosen. Intuitively, we can see a nice property in this voting scheme. If one highly polarized proposal exists, it is unlikely to convince two other voters to vote for it and be the winner in this group of three. The `median' proposal out of the three is going to win, so a highly polarized proposal doesn't stand a chance. Below, we show a simple proof demonstrating this comparison.

\begin{figure*}
    \centering
    \subfloat[][$f(x) = 1$]{
    \includegraphics[width=.3\textwidth]{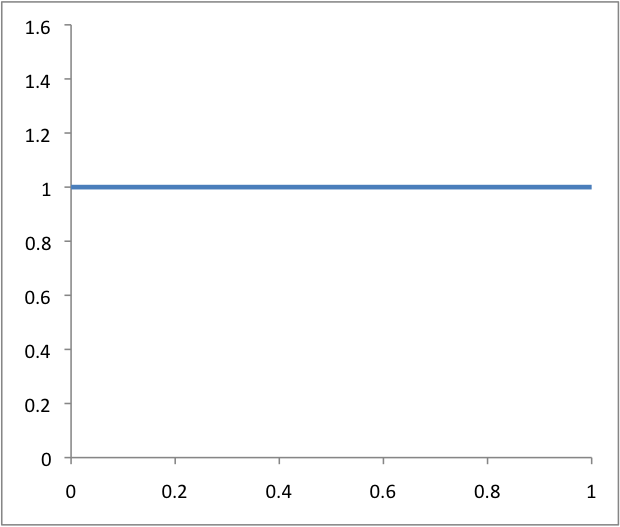}
    }
    \subfloat[][$g_{\text{Triadic}}$]{
    \includegraphics[width=.3\textwidth]{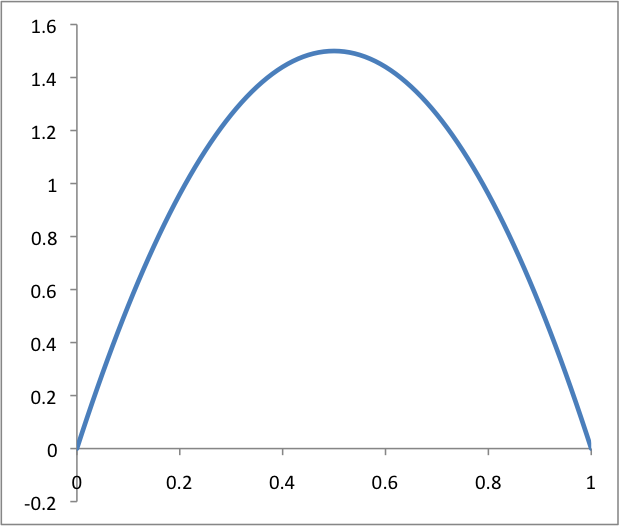}
    }
    \subfloat[][$g_{\text{Hot-or-Not}}$]{
    \includegraphics[width=.3\textwidth]{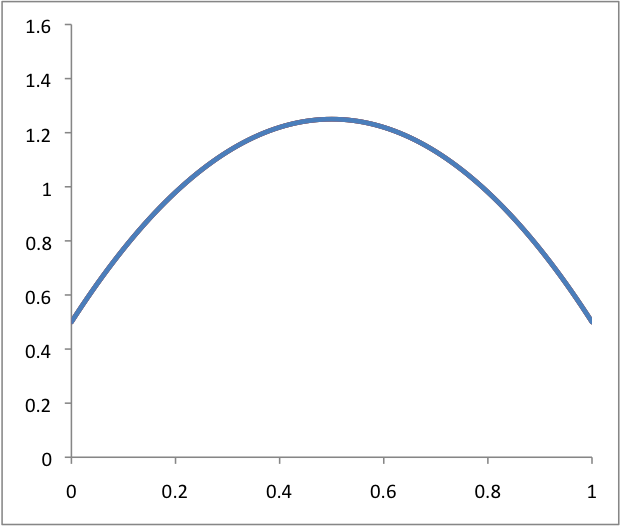}
    }
    \caption{(a) The original distribution of proposals, (b) The distribution of the proposal after one round of the Triadic voting scheme, (c) The distribution of the proposal after one round of the Hot-or-Not voting scheme}
    \label{fig:VotingScheme}
\end{figure*}

\begin{theorem}
Let a continuous distribution of voters be uniformly distributed between zero and one. Let $g_{\text{Hot-or-Not}}(x)$ and $g_{\text{Triadic}}(x)$ be the probability density of $x$ being the next winning candidate in the Hot-or-Not and Triadic voting protocols respectively. Then,
\begin{align*}
g_{\text{Triadic}} &= 6x(1-x)\\
g_{\text{Hot-or-Not}} &= 3x(1-x) + \frac{1}{2}
\end{align*}
In particular,
\[g_{\text{Hot-or-Not}} = \frac{1}{2}g_{\text{Triadic}}+\frac{1}{2}f(x)\]
where $f(x) = 1$ is the uniform distribution.
\end{theorem}
\begin{proof}
In a round of Triadic voting in a one-dimensional proposal space, the median point will always win. This is clear because the right-most point and the left-most point will vote for the median point. Regardless of who the median point votes for, he still wins with two votes. Then,
\[g_{\text{Triadic}} = 3!f(x)F(x)(1-F(x))\]
where $F(x)$ is the cumulative distribution function of $f(x)$. This is because $3!$ is the number of ways to choose three points, $f(x)$ is the probability of choosing point $x$, $F(x)$ is the probability of choosing a point left of $x$, and $1-F(x)$ is the probability of choosing a point right of $x$. So the above is the probability of $x$ being chosen as the median point, which is its probability of winning. For a uniform distribution of voters, $f(x) = 1$ and $F(x) = x$, so
\[g_{\text{Triadic}} = 6x(1-x)\]

In a round of Hot-or-Not voting in a one-dimensional proposal space, the point $x$ will win against $y$ only if the voter $z$ is closer to $x$ than $y$. Then,
\begin{align*}
    g_{\text{Hot-or-Not}} &= 2f(x)\int_0^x f(y)\left(1-F\left(\frac{x+y}{2}\right)\right)\mathrm{d}x\\
    &\ + 2f(x) \int_x^1 f(y)F\left(\frac{x+y}{2}\right)\mathrm{d}x 
\end{align*}
where we have a factor of $2$ instead of $6$ because it matters which point is the voter, so there are only two ways of choosing. For a uniform distribution of voters, we have
\begin{align*}
    g_{\text{Hot-or-Not}} &= 2\left[\int_0^x \left(1-\frac{x+y}{2}\right)\mathrm{d}x + \int_x^1 \left(\frac{x+y}{2}\right)\mathrm{d}x\right]\\
    &= 2\left[\left(1-\frac{x}{2}\right)y\Big|_0^x - \frac{1}{4}y^2 \Big|_0^x + \frac{x}{2}y\Big|_x^1 + \frac{1}{4}y^2\Big|_x^1\right]\\
    &= 2\left[\left(x-\frac{x^2}{2}\right) - \frac{1}{4}x^2 + \frac{x}{2} - \frac{1}{2}x^2 + \frac{1}{4} - \frac{1}{4}x^2\right]\\
    &= 3x(1-x) + \frac{1}{2}\\
    &= \frac{1}{2}g_{\text{Triadic}}+\frac{1}{2}f(x)
\end{align*}
\end{proof}

In the above theorem, we show the probability that any given candidate budget is not eliminated after one round of voting. The distributions are plotted in Figure~\ref{fig:VotingScheme}. It is quite clear that the Triadic voting scheme is significantly better than the Hot-or-Not scheme at eliminating proposals that are close to the boundary and centering the probability around the median candidate.

%
%
%

\section{User Trials - Reaching Consensus}

As mentioned in Section~\ref{sec:intro}, it is hard to find web platforms that facilitate discussions, especially when the issues are complex and contentious. For example, in the case of budgets, it is difficult to express views in standard fora due to the rich financial data coupled with the numerous proposals made by other people. It is especially hard to construct one's opinion or budget from the data while agreeing or disagreeing with others' proposals for different categories.

In this section, we ask the social implications of a web platform that facilitates discussions on complex issues. For example, we would like to know if users actually changed their mind after collaborating with others on budgets using the Widescope platform. In particular, we will also demonstrate the use of the triadic voting scheme, the properties of which were analyzed in the previous section.

\subsection{Experiment}

In order to answer the above question, we conducted a pilot user trial with twelve subjects as follows. First, we set a goal that can be quantified, i.e. we wanted users to reduce the federal budget deficit by half. Second, we gave our subjects a limited time to make decisions (two hours). This leads to a much more controlled experiment. We divided the total time into rounds. The major steps for the study were:
\begin{enumerate}
\item Setup: The users were classified as either conservative or liberal based on a survey with simple yes/no answers. We chose twelve candidates of which half were conservatives and the other half liberal.  
\item Round 1 (30 minutes):  Users are given a list of internet resources  (e.g. for learning about the US Budget and an explanation of how to use the platform Widescope for creating a budget. Users then put together their own proposal on how to solve the US budget (along with their own justifications).
\item Round 2 (20 minutes): Four groups of three people are formed randomly. In these groups, two belong to the conservative camp and one to the liberal camp or vice versa. In each group, each member gets five minutes to explain their position to the others within the same group. Then, each member votes for one of the other two proposals. We expected that the two users from the same camp would vote for each other and the user from the minority camp would vote for the more moderate proposal from the majority camp. In this way, we would have identified a strongly leaning conservative/liberal (SC/SL) and a moderately leaning conservative/liberal (MC/ML). This is our triadic voting. 
\item Round 3 (30 minutes): We then pair up SC and SL users, MC and ML users, and the minority conservative and liberal users. There are two of each of these pairs, with the minority conservative/liberal groups acting as controls. Now, we ask each of these pairs to collaborate and present one consensus budget (a budget that both agree on) that cuts the deficit by at least 50\% or as much as they are able to.
\item Surveys: We presented them with survey questions at the end of each round about the following: 1) strength of opinions, 2) opinions on individual income taxes, 3) opinion on entitlement funding, 4) opinion on corporate income taxes, 5) opinion on national defense. At the end of Round 3, we also ask them how successfully they felt like they were able to come to consensus.
\end{enumerate}

\subsection{Results}

Our first observation was that in round 2, users did not vote according to our expectations. Often, a user from the majority camp voted for the proposal given by the minority member i.e. in a group with two liberals and a conservative, the liberal actually chose the conservative's budget! In fact 5 out of the 8 majority members voted for the opposing camp. 

We observed that many did not know a lot about the US budget, based on the basic nature of the questions we were asked during the trial. We conjecture that many citizens do not know  a lot of the details of the US budget and their opinions might be derived from opinions of friends. However, we believe, that once they are given access to educated information, they will make decisions based on the information and surprisingly, they do not just stick to their original position as is often unfortunately observed in Internet forums (and on capitol hill!)

We observed that users were highly successful at coming to consensus in the third round. 5 out of 6 groups came to consensus solutions within the time frame. The last group felt that they could have come to consensus if they had a little more time. 10 people rated 5 for how strongly they were able to come to consensus, 2 people rated 4. 10 people said that given more time, they would come to a 50\% consensus, 2 people said the other person was reasonable but had fundamentally different views. This was a big surprise to us, and is the key empirical finding of this paper.

We noticed that several users had stronger opinions by the end of the user trial.  Average strength of opinions (1-5) changed from 2.75 to 3.5 to 3.67 from pre-trial, post voting (Round 2), and post consensus (Round 3) stages respectively. Looking at it in more detail, 7 users had a stronger opinion by the end of the user trial, 4 users had the same strength of opinion, 1 user had a weaker opinion (i.e. her score went to 2 from a 3). 
We believe that after people have become more educated, and have obtained more information, they have a stronger opinion since they feel it is backed up with evidence.

Another very interesting observation is that people actually flipped sides for several proposals during the trial! Looking at each of the opinions (from the surveys), we had that: 
1) Individual taxes: 7 had the same score before and after, 3 mid ranked users (with a score of 3) moved to moderate positions on both sides (a score of either 2 or 4), 1 user become more strongly opinionated (increased score from 4 to 5), 1 user went from  moderate to neutral (score 3 changed to 2).
2) Entitlement funding: 3 had the same score before and after, 4 became stronger in their original position (i.e. either increased or decreased their scores further), 2 moved from moderate positions to neutral, 1 moved from neutral to moderate, and 2 flipped sides completely! (e.g. score changed from 5 to 2, or score changed from 2 to 4).
3) National Defense: 4 had the same position before and after, 3 become stronger in their opinion (score decreased ffrom 2 to 1), 3 flipped sides (score of 4 became 5 and score of 2 became 1), 1 became neutral (score 4 decreased to 3), 1 became more moderate (score changed from 5 to 4). In the case of national defense, all the movement was one sided towards the direction of decreasing military spending. This is surprising given that we had an even mix of liberals and conservatives.
4) Corporate taxes: 6 had the same position before and after, 1 moved to neutral (score decreased from 4 to 3), 1 moved to more strong (score changed from 4 to 5), 2 moved away from neutral to strong (score 3 increased to 5), 1 moved to weaker (score 5 decreased to 4), 1 flipped sides (score 4 became 2).
It is especially amazing to see that people flipped sides in some topics (in both directions). It is also interesting that for certain areas like national defense, every person moved in one direction (this is also because both conservative and liberal parties believe that military costs should be cut, just in differing amounts). 

Thus, our observations confirm our hypothesis that tools like widescope might be able to help users to formulate their opinions and have a conversation about a complex subject. Also it was very interesting to note that our tool helped users to change their opinions upon analyzing the rich financial data, instead of forming opinions based on qualitative beliefs. Also, contrary to our expectation, we found that almost all groups of conservative-liberal pairs were able to come to consensus through using our website as a collaboration tool. This is a result that gives hope to a problem that didn’t seem solvable and points to many future questions in collaborative website design. 

\subsubsection{Discussions}

During our user trial, we found the users to be generally reasonable. Here are a couple of further observations:
\begin{enumerate}
\item Widescope gives people good common data they can research and have discussion based on a common understanding. Typical websites for conflict-resolution do not have this data, so users can only resort to online  shouting matches to get their voices heard.
\item Widescope captures the entire opinion through one concept. Typical websites like online fora are unstructured. As a result, people cannot observe the effects of their opinions. Often they go back and forth on different issues i.e. a user may have many good reasons for cutting taxes, but when they do so, they will have to see that cutting taxes by x\% will make it so much harder to decrease the deficit. However with widescope they can see the effect visually and they're forced to consider the budget as a whole. 
\item Widescope helps users be more rational. When a user proposes a change in a budget subcategory, the system displays a variety of proposals to choose from, thus giving the user a chance to rationalize. Thus our platform using subtle cues to help be more rational and consider other's proposals. 
\item Widescope allows us to have a {\em goal} based consensus via collaboration. On most websites, the goal is more to exert one's opinion. Here, if the other person disagrees, then it will be evident during a collaboration (via the disagreement meter). 
As a side note, this also points to aiding the solution of real problems with average citizens. Though the budget committee and congress members are more informed and spend their complete attention on these issues, they might have their own biases. Average citizens, though less informed, as an aggregate, may have  fewer biases and self interests. Thus widescope can be viewed as a platform for crowdsourcing of ideas on complex issues. 
\end{enumerate}

Media is often seen to polarize people. People often pick out the information that supports their viewpoint, etc. However, in our user trial, we saw a positive effect of leveraging the web. We found that people become more strongly opinionated when they receive more information, but they do not necessarily need to be {\em polarized} by it. When placed in the context of collaborative discussion, rather than being {\em talked at} by the media, results can go both ways. 

\section{Future Directions}

As we mentioned before (Section~\ref{sec:intro}), we have touched the tip of the iceberg in designing systems that enable serious conversations on complex data driven topics via consensus and collaboration. It is a new area and is currently being investigated in the following directions:

{\bf UI improvements:} We are continuing to improve the UI. The collaboration tool is being enhanced to include the status of users and suggestions of users one might want collborate with. We are also prioritizing and ranking the list of the proposals that are displayed when users change their own proposals. We are also introducing micro proposals. 

{\bf Large scale user trials and experiments:} Given the surprises in the small user trial described above, we are inspired to try bigger trials with the tools we have developed. While we can collaborate in real time using our tool, we need to build another tool to manage the collaboration space. Also we need to highlight and rank successful collaborations. We are considering building an online  campaign that uses the components of Widescope to help users converge on budget issues. 

{\bf Theoretical understanding:} An interesting question is to design more efficient voting schemes for faster convergence towards consensus proposals. For example how many triadic groups do we need to sample from a population for us to reach {\em approximate consensus}. Also, there might be other better voting mechanisms apart from our triadic voting. 

{\bf Link with Social Choice and Behavioral Economics:} We believe that further research and user trials will lead to deeper links between novel visualization, social choice theory, and behavioral economics ~\cite{behavioralecon}. For example, what kind of newer agent models might arise by leveraging Widescope for collborative debates on complex subjects and how that might impact the aggregate choices.

\section{Conclusions}
\label{sec:conclusions}

In this paper we presented a web platform called Widescope and showed it can be useful for enabling conversations around complex (and data driven) topics such as budgets using rich data visualization, ability to define one's own budget while choosing between others' proposals or creating one's own, ability to compare between budgets and collaborate in real time with other users. In addition, we showed that this platform has the potential to make social impact by helping users converge their thoughts through the lens of a user trial. 

We  modeled the triadic voting scheme for proposals (where a proposal mapped onto a  single-dimension point) and showed that it is better than the comparitive hot-or-not voting since it centers the probability of winners around those who have moderate proposals, thus satisfying the Condorcet property. This simple formal result is also demonstrated in our user trial.

\section{Acknowledgments}
The authors would like to acknowledge several members of the Widescope team, 
especially Eric Mibuari. Also we would like to acknowledge those who provided 
extremely valuable feedback like Tim O'Reilley, Amir Efrati. This research was supported by Ashish Goel's NSF grants 0904325 and 0947670; by in-kind support from Amazon; and by funds from Cisco, MIcrosoft, and Google.

%
\bibliographystyle{plain}
\bibliography{widescope}  

\begin{thebibliography}{10}

\bibitem{ec2}
Amazon ec2.
\newblock http://aws.amazon.com/ec2/.

\bibitem{facebook}
Facebook.
\newblock http://www.facebook.com.

\bibitem{federalbudget}
The federal budget.
\newblock http://www.federalbudget.com/.

\bibitem{googleplus}
Google plus.
\newblock http://plus.google.com.

\bibitem{hotornot}
Hot-or-not.
\newblock http://www.hotornot.com.

\bibitem{jquery}
Jquery.
\newblock http://jquery.com.

\bibitem{mysql}
mysql.
\newblock http://www.mysql.com.

\bibitem{protoviz}
Protoviz.
\newblock http://mbostock.github.com/protovis/.

\bibitem{realhotornot}
Real hot-or-not.
\newblock http://www.facebook.com/apps/application.php?id=91683468488.

\bibitem{ror}
Ruby on rails.
\newblock http://rubyonrails.org/.

\bibitem{twitter}
Twitter.
\newblock http://www.twitter.com.

\bibitem{widescope}
Widescope.
\newblock http://widescope.stanford.edu.

\bibitem{behavioralecon}
Congressional~Progressive Caucus.
\newblock The people's budget.
\newblock http://en.wikipedia.org/wiki/Behavioral\_economics.

\bibitem{progress}
Congressional~Progressive Caucus.
\newblock The people’s budget.
\newblock http://cpc.grijalva.house.gov/index.cfm?sectionid=70.

\bibitem{wsjblurb}
Amir Efrati.
\newblock Crowdsourcing deficit reduction.
\newblock http://on.wsj.com/n58wMe.

\bibitem{ui:voyager}
Jeffrey Heer, Fernanda~B. Vi{\'e}gas, and Martin Wattenberg.
\newblock Voyagers and voyeurs: Supporting asynchronous collaborative
  visualization.
\newblock {\em Commun. ACM}, 52:87--97, January 2009.

\bibitem{washbudget}
Michael Itti.
\newblock Washington state budget: What are your priorities?
\newblock http://www.educationvoters.org/session-2011/budget/.

\bibitem{knox62}
Robert~E. Knox and Paul~J. Hoffman.
\newblock Effects of variation of profile format on intelligence and
  sociability judgments.
\newblock {\em Journal of Applied Psychology}, 46(1):14--20, 1962.

\bibitem{ui:pathfinder}
Kurt Luther, Scott Counts, Kristin~B. Stecher, Aaron Hoff, and Paul Johns.
\newblock Pathfinder: an online collaboration environment for citizen
  scientists.
\newblock In {\em Proceedings of the 27th international conference on Human
  factors in computing systems}, CHI '09, pages 239--248, New York, NY, USA,
  2009. ACM.

\bibitem{nytimesbudgetpuzzle}
NewYorkTimes.
\newblock Budget puzzle: You fix the budget.
\newblock
  http://www.nytimes.com/interactive/2010/11/13/weekinreview/deficits-graphic.%
html.

\bibitem{budgetchallenge}
NextTen.
\newblock California budget challenge.
\newblock http://budgetchallenge.org.

\bibitem{timespent}
Nielsen.
\newblock What americans do online: Social media and games dominate activity.
\newblock http://goo.gl/CR6f.

\bibitem{pres_budget}
Office of~Management and Budget.
\newblock The president’s budget for fiscal year 2012.
\newblock http://www.whitehouse.gov/omb/budget.

\bibitem{prosperity}
House of~Representatives Committee on~the Budget.
\newblock The path to prosperity: Restoring america's promise.
\newblock http://budget.house.gov/fy2012budget/.

\bibitem{cbo_options}
Congressional~Budget Office.
\newblock Reducing the deficit: Spending and revenue options.
\newblock http://www.cbo.gov/doc.cfm?index=12085.

\bibitem{socialchoice}
Wikipedia.
\newblock Social choice theory.
\newblock http://en.wikipedia.org/wiki/Social\_choice\_theory.

\bibitem{voting}
Wikipedia.
\newblock Voting.
\newblock http://en.wikipedia.org/wiki/Voting\_theory.

\bibitem{condorcet}
Wikipedia.
\newblock Voting.
\newblock http://en.wikipedia.org/wiki/Condorcet\_criterion.

\bibitem{ui:commentspace}
Wesley Willett, Jeffrey Heer, Joseph Hellerstein, and Maneesh Agrawala.
\newblock Commentspace: structured support for collaborative visual analysis.
\newblock In {\em Proceedings of the 2011 annual conference on Human factors in
  computing systems}, CHI '11, pages 3131--3140, New York, NY, USA, 2011. ACM.

\end{thebibliography}
%
%

\end{document}